\documentclass[11pt]{article}

\usepackage[letterpaper,margin=1in]{geometry}
\usepackage[T1]{fontenc}
\usepackage{amsmath,amsthm}
\usepackage{newtxtext,newtxmath}
\usepackage{microtype}
\usepackage[lined,ruled,linesnumbered]{algorithm2e}
\usepackage{xcolor}
\usepackage{hyperref}
\usepackage{cleveref}
\usepackage{tikz}
\usepackage{pgfplots}
\pgfplotsset{compat=1.17}
\usepackage[backend=biber,style=alphabetic,maxnames=8,maxalphanames=4,minalphanames=3]{biblatex}
\DeclareDelimFormat{multicitedelim}{\addcomma\space}
\hypersetup{
  colorlinks=true,
  linkcolor=blue!55!black,
  citecolor=blue!55!black,
  urlcolor=blue!55!black,
  hypertexnames=false,
  pdfauthor={},
  pdftitle={Algorithmic List Decoding of Reed--Solomon Codes up to Capacity in the Low-Rate Regime}
}

\allowdisplaybreaks

\newtheorem{theorem}{Theorem}[section]

\newtheorem{lemma}[theorem]{Lemma}
\newtheorem{proposition}[theorem]{Proposition}
\newtheorem{corollary}[theorem]{Corollary}

\newcommand{\rank}{\operatorname{rank}}

\newcommand{\F}{\mathbb F}
\newcommand{\E}{\mathbb E}

\newcommand{\eps}{\varepsilon}
\let \cref \Cref

\title{\textbf{Algorithmic List Decoding of Reed--Solomon Codes\\
up to Capacity}}
\author{
Joshua Brakensiek\thanks{
University of California, Berkeley.
\href{mailto:josh.brakensiek@berkeley.edu}{\texttt{josh.brakensiek@berkeley.edu}}
Supported in part by a Simons Investigator award of Venkatesan Guruswami, and NSF awards CCF-2211972 and DMS-2503280.
}
\and
Yeyuan Chen\thanks{
Department of EECS, University of Michigan, Ann Arbor.
\href{mailto:yeyuanch@umich.edu}{\texttt{yeyuanch@umich.edu}}
Supported in part by NSF award CCF-2236931.
}
\and
Aaron Putterman\thanks{
Harvard University, Cambridge, Massachusetts.
\href{mailto:aputterman@g.harvard.edu}{\texttt{aputterman@g.harvard.edu}}
Supported in part by a Jane Street Graduate Research Fellowship, the Simons Investigator awards of Madhu Sudan and Salil Vadhan, and AFOSR award FA9550-25-1-0112.
}
\and
Zihan Zhang\thanks{
Simons Institute for the Theory of Computing, Berkeley and
Institute for Advanced Study, Princeton.
\href{mailto:zzhsdj@foxmail.com}{\texttt{zzhsdj@foxmail.com}}
}
\and
Kai Zhe Zheng\thanks{
Simons Institute for the Theory of Computing, Berkeley and
Institute for Advanced Study, Princeton.
\href{mailto:kzzheng@mit.edu}{\texttt{kzzheng@mit.edu}}
}
}
\date{}

\begin{document}
\maketitle

\begin{abstract}
We provide a deterministic polynomial-time list decoding algorithm for Reed--Solomon codes over prime fields that approaches list decoding capacity on every evaluation set for any (constant) rate. 
\end{abstract}

\section{Introduction}
In the field of coding theory, the most fundamental error-correcting code is the Reed--Solomon (RS) code~\cite{ReedSolomon1960} which has found numerous practical applications including in data storage~\cite{pursley1985performance,plank1997tutorial, plank2009performance, guruswami2016repairing}, wireless communication~\cite{wu1987coding, wicker1992reed, koetter2003algebraic, gross2006applications}, and cryptography~\cite{shamir1979share, mceliece1981sharing, chaum1988multiparty, rabin1989verifiable, pedersen1991threshold, franklin1992communication, juels2006fuzzy, goldberg2007improving, shacham2008compact, ben2019completeness,BBHR18,ACFY24,ACFY25}.
Simply stated, a Reed--Solomon interprets a message to be transmitted as the coefficients of a low-degree univariate polynomial which is then evaluated at many evaluation points to introduce redundancy. More formally, given a finite field $\F_q$ and distinct evaluation points $\alpha_1,\ldots,\alpha_n \in \F_q$, one can construct a Reed--Solomon code of block length $n$ and rate $R:=k/n$ as follows:
\[
\mathsf{RS}_{n, k}(\alpha_1,\ldots,\alpha_n)
=\left\{(P(\alpha_1),\ldots,P(\alpha_n))\in\F_q^n: P\in\F_q[X]\text{ }\text{ and }\deg P<k\right\}.
\]
A defining characteristic of Reed--Solomon codes is that they are maximum distance separable (MDS). In other words, every nonzero codeword $\vec{c} \in \mathsf{RS}_{n,k}(\alpha_1, \hdots, \alpha_n)$ has Hamming weight at least $n-k+1$. Equivalently, the code has relative distance $\delta=\frac{n-k+1}{n}$, perfectly meeting the Singleton bound~\cite{singleton1964maximum}. From the perspective of error-correction, where we receive a message $\vec{y} \in \F_q^n$, Reed--Solomon codes have optimal \emph{unique decoding}. Concretely, if there exists a codeword $\vec{c} \in \mathsf{RS}_{n,k}(\alpha_1, \hdots, \alpha_n)$ with Hamming distance at most $\lfloor \frac{n-k}{2}\rfloor$ (approximately $(1-R)/2$ relative distance) from $\vec{y}$, then $\vec{c}$ is unique and in fact can be found efficiently~\cite{peterson1960encoding,gorenstein1961class,Ber68b,Mas69}.

A much more ambitious quest by the coding theory community has been to understand the decodability of Reed--Solomon codes beyond the unique decoding regime, namely in the paradigm of \emph{list decoding}, introduced independently by Elias~\cite{elias1957list} and Wozencraft~\cite{wozencraft1958list}. In this setting, we seek to find a list of codewords which are within some specified Hamming distance of our corrupted message. By relaxing this notion of recovery, a much larger fraction of errors can be recovered from~\cite{Sudan1997,GuruswamiSudan1999}. The concept of list decoding has proved to be remarkably versatile, finding applications across the whole spectrum of theoretical computer science such as pseudorandomness and randomness extraction
\cite{SudanTrevisanVadhan01,Trevisan01,GuruswamiUmansVadhan09},
average-case complexity and hardness amplification
\cite{GoldreichLevin89,GoldreichRubinfeldSudan00},
cryptography and related combinatorial problems
\cite{SilverbergStaddonWalker01,AkaviaGoldwasserSafra03,
JuelsSudan06,AlonGuruswamiKaufmanSudan07}, and interactive proof systems \cite{BenSassonCarmonIshaiKoppartySaraf23,BenSassonCarmonHaboeckKoppartySaraf26}.

We say that a Reed--Solomon code is $(\rho, L)$ list decodable if for every possible message $y \in \F_q^n$, there are at most $L$ codewords which have relative Hamming distance at most $\rho$ from $y$. Unlike unique decoding which fails to exist beyond a radius of $(1-R)/2$, list decoding is a meaningful notion up to a radius of $1-R$, known as \emph{list-decoding capacity}~\cite[Theorem~7.4.1]{guruswami2012essential} --- beyond a radius of $1-R$, a list of size exponential in $n$ is always required.  More precisely, we say that our Reed-Solomon code attains list-decoding capacity if for any $\eps > 0$, our code is $(1-R-\eps, n^{O_{\eps}(1)})$ list decodable. A parameter regime of particular interest is the \emph{low-rate regime} (or \emph{high-noise regime}) where rate $R = \Theta(\eps)$ and (relative) radius $\rho = 1-\eps$.

In general, it is not understood which Reed--Solomon codes achieve list-decoding capacity. For the past three decades, the benchmark for the Reed--Solomon and other classes of codes has been the \emph{Johnson radius} of $1-\sqrt{R}$, for which the seminal works of Sudan~\cite{Sudan1997} and Guruswami and Sudan~\cite{GuruswamiSudan1999} gave polynomial-time list-decoding algorithms\footnote{Recent work of Chatterjee, Harsha, and Kumar~\cite{ChatterjeeHarshaKumar2026} gives a deterministic polynomial-time list decoding algorithm up to the Johnson radius.} for RS codes, culminating in efficient decoding up to the \emph{Johnson radius} $1-\sqrt{R}$
up to lower-order terms. This falls short of the list-decoding capacity $1-R$, leaving a substantial gap between what is information-theoretically possible and what is known to be efficiently achievable for RS codes.

The Johnson radius has served as a hard barrier for progress in the field for good reason. Several works provided evidence of barriers to list decoding RS codes substantially beyond the Johnson radius~\cite{RR03,GR06,CW07,BKR10}. In particular, Ben-Sasson, Kopparty, and Radhakrishnan~\cite{BKR10} exhibited, over fields of bounded characteristic, RS codes with superpolynomially large lists at decoding radius approaching the Johnson bound in certain parameter regimes. Complementing these combinatorial obstructions, several works established computational hardness for RS decoding at substantially larger decoding radius~\cite{GV05,CW07,GGG18}. On the positive side, a recent line of work has shown that RS codes with random evaluation points achieve list-decoding capacity with the optimal list size~\cite{RudraWootters2015,ShangguanTamo2023,GuoLiShangguanTamoWootters2021,GoldbergShangguanTamo2023,BrakensiekGopiMakam2024,GuoZhang2023,AlrabiahGuruswamiLi2024,AlrabiahGuoGuruswamiLiZhang2025}.\footnote{These are combinatorial results: they bound the number of codewords in a Hamming ball, without providing an efficient algorithm for finding them.} In fact, stronger results \cite{BrakensiekGopiMakam2024,GuoZhang2023,AlrabiahGuruswamiLi2024,AlrabiahGuoGuruswamiLiZhang2025} attaining the generalized Singleton bound of Shangguan and Tamo \cite{ShangguanTamo2023} are known for every fixed list size. In contrast, no explicit family of RS codes is known to achieve comparable combinatorial guarantees. Moreover, even for the RS codes  with random evaluation points covered by these results, no efficient algorithm is known to attain their beyond-Johnson decoding radius.

In part due to the lack of progress on list-decoding of RS codes beyond the Johnson codes, many coding theorists considered variants of RS codes such as Parvaresh--Vardy codes~\cite{parvaresh2005correcting}, Folded Reed-Solomon Codes~\cite{GuruswamiRudra2008}, and Univariate Multiplicity codes~\cite{kopparty2014multiplicity}. In particular, folded Reed--Solomon codes introduced by Gurusawmi and Rudra~\cite{GuruswamiRudra2008} are the first explicit family of codes that can be efficiently list decoded up to capacity, albeit with a field size $n^{\Theta(1/\eps)}$ when $\eps$-close to capacity.

A subsequent line of work further developed and strengthened the folded Reed--Solomon framework, improving various aspects of its list-decoding guarantees~\cite{GuruswamiWang2013,Kopparty2015,KoppartyRonZewiSarafWootters2023,GHKS24,srivastava2025improved,chen2025explicit,ashvinkumar2026algorithmic}. Related ideas have also led to several other explicit capacity-achieving algebraic codes such as~\cite{GRZ21,GX22}. 

Despite this remarkable recent progress, the following fundamental question, dating back to the work of Guruswami and Sudan, remains open.
\begin{center}
{\centering
\itshape
Can Reed--Solomon codes of rate $R$ be efficiently list decoded
beyond the Johnson radius $1-\sqrt{R}$?
}
\end{center}
In fact, even an explicit Reed--Solomon code with combinatorial list-decoding bound beyond the Johnson radius remained unknown \footnote{Very recently, we learned through personal communication of partial derandomization results attaining the generalized Singleton bound for a certain constant list size, albeit with a super-exponential alphabet size~\cite{ChatterjeeEtAl2026ExplicitRS}.}.

In this paper, we resolve this question in the affirmative for any constant rates over prime fields. Specifically, we give a polynomial-time algorithm that list decodes Reed--Solomon codes substantially beyond the Johnson radius (even achieving list decoding capacity!). The result holds for arbitrary sets of distinct evaluation points, with no randomness assumption. 
The crucial step in our argument is to establish the result in a sufficiently low, but still constant, rate regime. A simplified statement of our main result in this regime is given by the following informal theorem (\Cref{thm:main}).
\begin{theorem}[Main Result, Informal Version in Low (Constant) Rate Regime]\label{thm:main}
Fix a constant parameter $\theta\in(0,1).$ Then there exists a constant $\eps_0(\theta)>0$ such that for every constant $\eps$ with $\eps_0(\theta)>\eps>0$, any Reed--Solomon code
with arbitrary distinct evaluation points over a prime field
$\mathbb{F}_q$ of size $q=\Theta_{\theta,\eps}(n)$ and rate
\[
R\leq(1-\theta)\eps
\]
can be efficiently list decoded from up to $1-\epsilon$ fraction of errors,
with list size bounded by $n^{\operatorname{poly}_{\theta}(1/\epsilon)}$.
\end{theorem}

We remark that the Johnson radius only has rate $O(\epsilon^2)$ with respect to the radius $1-\epsilon$. Below, we discuss the origins of our new ideas and provide an overview of the techniques. See \cref{thm:main-formal} for a more precise asymptotic statement. It is notable to mention for \cref{thm:main} that in the regime for which $n \gg \Theta_{\theta,\eps}(n/k)$, we can actually select $q = n$. In particular, for small constant rates, our algorithm works for Reed--Solomon codes where the evaluation points are precisely $\F_q$ itself.

\paragraph{Extension to All Rates.}
 Subsequent to the first posting of our paper, thanks to a simple reduction due to Omar Alrabiah, Rohan Goyal, and Venkatesan Guruswami, this low-rate result then extends to any constant rate, as stated in the subsequent corollary (\Cref{cor:all-rates}). The proof of \Cref{cor:all-rates} is included in \Cref{sec:main-cor}, along with a more formal statement in \cref{cor:all-rates-formal}.
 
\begin{corollary}[List Decoding up to Capacity at All Constant Rates, Informal Version]
\label{cor:all-rates}
Fix constants $R\in(0,1)$ and $\delta\in(0,1-R)$.
Then there exists a constant $C=C(R,\delta)>0$ such that, for all
sufficiently large $n$, every Reed--Solomon code of block length $n$ and
rate at most $R$, with arbitrary distinct evaluation points over a prime
field $\F_q$ of size $q\ge Cn$
can be efficiently list decoded from a
\[
1-R-\delta
\]
fraction of errors.  Moreover, the output list has size
$n^{O_{R,\delta}(1)}$.
\end{corollary}

\paragraph{Further Developments.} In addition to the aforementioned black-box proof of \cref{cor:all-rates}, subsequent to the initial posting of our paper proving \cref{thm:main}, a number of groups contemporaneously realized (with LLM assistance) that the proof of \cref{thm:main} can be modified in a white-box manner to achieve capacity for all constant rates. Such developments will be documented in more detail in a future version of this manuscript.

\subsection{Origins of the Main Ideas and Technique Overview}
In the list-decoding problem, we are given the evaluation points $(\alpha_i \in \F_q)_{i=1}^n$, the received word $(y_i \in \F_q)^n_{i=1}$, the target agreement $A$, and the code dimension $k$. The goal is to find all polynomials $P\in\F_q[X]_{\leq k-1}$ (i.e., with degree at most $k-1$) such that $\sum^n_{i=1}[P(\alpha_i)=y_i]\ge A$ in polynomial time, where $A$ should be as small as possible. Our algorithm still follows the high level template introduced by Sudan \cite{Sudan1997} and Guruswami--Sudan \cite{GuruswamiSudan1999} consisting of interpolation and root-finding. In our setting however, the object that we interpolate is significantly different. 
Indeed, rather than constructing an interpolating polynomial involving only the
unknown message polynomial $P(X)$, we introduce several of its \emph{Hasse
derivatives} as additional formal variables.
 
Recall that the Hasse derivatives are the coefficients in the formal Taylor expansion
\[
P(X+T)=\sum_{j\geq 0}P^{[j]}(X)T^j.
\]
Crucially, however, the decoder is \emph{not} given any derivative information
about the unknown polynomial $P$. At each evaluation point $\alpha_i$, the
only information available to us is the ordinary received value $y_i$, and at
an agreement position we know only that
\[
    P(\alpha_i)=y_i.
\]
In particular, the values
\[
    P^{[1]}(\alpha_i),\ldots,P^{[d]}(\alpha_i)
\]
are completely unknown. A central challenge is therefore to impose
interpolation constraints involving these higher-order derivatives using only
the zeroth-order agreement condition $P(\alpha_i)=y_i$.
The central idea is to construct a multivariate polynomial $Q$ which vanishes when evaluated on the tuple $\bigl(P(X),P^{[1]}(X),\ldots,P^{[d]}(X)\bigr)$ for every polynomial $P$ which has sufficiently high agreement with the received word. With this context, our algorithm consists of an interpolation procedure and a root-finding procedure.

\begin{itemize}
    \item[1.] {Interpolation Step}. Find a \emph{non-zero} polynomial $Q\in\F_q[X,Y_0,Y_1,\dots,Y_d]$ with $(1,k-1,\dots,k-d-1)$-degree less than $mA$, where $d$ represents the highest-order Hasse derivative we consider, and $m$ is our target root multiplicity. In other words, for every monomial $X^aY_0^{j_0}Y_1^{j_1}Y_2^{j_2}\cdots Y_d^{j_d}$ appearing in $Q$, we have that $a+(k-1)j_0+(k-2)j_1+\cdots+(k-d-1)j_d < mA$.
    This property ensures that for any polynomial $P \in \F_q[X]$ of degree less than $k$, we have that $Q(X, P(X), P^{[1]}(X), \hdots, P^{[d]}(X))$ is a univariate polynomial of degree less than $mA$.
    We also require that $Q$ satisfies the following additional constraint: If $P\in\F_q[X]_{\leq k-1}$ satisfies $P(\alpha_i)=y_i$, then $(X-\alpha_i)^m$ is a factor of $Q(X,P(X),\dots,P^{[d]}(X))$. The degree restrictions on $Q$ are chosen so that, for every polynomial
$P\in\F_q[X]_{\le k-1}$,
\[
\deg_X Q\bigl(X,P(X),P^{[1]}(X),\ldots,P^{[d]}(X)\bigr)<mA.
\]
Consequently, if $P$ agrees with the received word in at least $A$ positions,
then the multiplicity condition above gives at least $mA$ roots of this polynomial counted with
multiplicity, which guarantees
\begin{equation} \label{eq:identical-vanish}
Q\bigl(X,P(X),P^{[1]}(X),\ldots,P^{[d]}(X)\bigr)\equiv0.
\end{equation}
We call this space of the polynomials satisfying such degree restrictions our interpolation space and henceforth denote it by $\mathcal Q$. For technical reasons, it will also be helpful to isolate the allowed monomials involving only $Y_2,\ldots, Y_d$ and denote them by $\mathcal{B}$.
    
\item[2.] Root-finding Step. Having found $Q$ as above, the guarantee in \eqref{eq:identical-vanish} reduces our list-decoding problem to a root-finding one. Namely, we must find all low degree $P\in\F_q[X]_{\leq k-1}$ which satisfy \eqref{eq:identical-vanish}. Fortunately, this can be done by using the Univariate Multiplicity decoder of Kopparty~\cite{Kopparty2015} essentially as a blackbox. Altogether, we are guaranteed to find a list containing every polynomial $P$ with agreement at least $A$. Since this procedure may also find spurious polynomials, we check that for each such recovered $P$ whether it is indeed a solution with at least $A$ agreements with the received word.
\end{itemize}

\paragraph{Finding the Prescribed $Q$.}
Given the above outline, one can see that the main difficulty is therefore the interpolation step: how can we construct a nonzero $Q$ satisfying the required multiplicity condition without knowing the candidate polynomial $P$?

Fix one evaluation point $\alpha_i$ and received value $y_i \in \F_q$. When $P\in\F_q[X]$ satisfies $P(\alpha_i)=y_{i}$, we require $(X-\alpha_i)^m\mid Q(X,P(X),\dots,P^{[d]}(X))$, which is equivalent to 
\begin{equation}\label{eq:conditionE}
Q\left(\alpha_i+T,P(\alpha_i+T),P^{[1]}(\alpha_i+T),\dots,P^{[d]}(\alpha_i+T)\right) \equiv 0 \mod T^m.
\end{equation}
Notice that there is already an issue here. We do not know any of the values
\[
P(\alpha_i+T),P^{[1]}(\alpha_i+T),\ldots,P^{[d]}(\alpha_i+T).
\]
A naive approach is to introduce formal variables $Y_0,\ldots, Y_d$ for the respective derivatives above, and simply require 
\begin{equation}\label{eq:conditionF}
Q\left(\alpha_i+T,Y_0,\ldots,Y_d\right)
\equiv 0 \pmod{T^m},
\end{equation}
over every formal choice of $Y_0, \ldots, Y_d$ where the congruence is in $\F_q[T,Y_0,\ldots,Y_d]$. Then, to find our desired interpolant $Q$, we set up a homogeneous linear system in variables corresponding to the monomials in $\mathcal{Q}$ with one equation per constraint imposed. Unfortunately, this approach leads to more constraints than monomials and thus has no guarantee of a non-zero solution.

\paragraph{Reducing the Interpolation Constraints.} Observe that such a requirement is far stronger than
necessary, however. Indeed, we only need the condition to hold for tuples that can arise from
a polynomial passing through $(\alpha_i,y_i)$. The key observation is that these tuples are not actually arbitrary. By the backward
Taylor identity, they satisfy the relation
\begin{equation}\label{eq:E-def}
Y_0
=
y_i+\sum_{j=1}^d(-1)^{j+1}T^jY_j+T E.
\end{equation}
where $E$ is a remainder term divisible by $T^d$. 

Thus, our first improvement over the naive method above is to no longer allow
$Y_0$ to vary independently of $Y_1,\ldots,Y_d$. Instead, we substitute
$Y_0$ according to the Taylor relation \eqref{eq:E-def}, while treating $Y_1,\ldots,Y_d$ and $E$ as free formal variables.

After this substitution, we expand
\[
Q\left(
\alpha_i+T,\,
y_i+\sum_{j=1}^d(-1)^{j+1}T^jY_j+TE,\,
Y_1,\ldots,Y_d
\right) =
\sum_{b,\vec{e}}q_{b,\vec{e}}(T)E^bY^{\vec{e}}.
\]
we impose the constraints
\begin{equation}
q_{b,\vec{e}}(T)\equiv0\pmod{T^{m-db}}
\end{equation}
for all $b$ such that $bd < m$. Note the reason we have different moduli per exponent $b$ is that the term $E^b$ supposedly contributes a factor of $T^{db}$ already, while we need the overall term $q_{b,\vec{e}}(T)E^bY^{\vec{e}}$ to vanish modulo $T^m$.

These constraints are imposed at all received points $\alpha_i$, for $i \in [n]$. The substitution performed is what allows us to only impose the vanishing condition at a more restricted class of structured points consistent with genuine successive derivatives.

Finally, to carefully bound the number of linearly independent constraints
imposed in this way, we view the local constraints at each received point as
a linear map on the interpolation space and upper bound its rank.  We do this by finding a large dimensional subspace in the kernel of this map and then performing a careful counting argument involving weighted lattice-points. We remark that there is a more direct way to count the number of linearly independent constraints, but we include the current approach as it reveals non-trivial structure about the kernel of the constraint map and could potentially be useful towards extending our results to higher rates or improving other constant dependencies.

\paragraph{Comparison with Guruswami--Sudan \cite{GuruswamiSudan1999}.}The Guruswami--Sudan algorithm \cite{GuruswamiSudan1999} focuses on  (\ref{eq:conditionE}) in the case $d=0$, but they cannot express $P(\alpha_i+T)$, either. \cite{GuruswamiSudan1999} instead imposes a stronger condition that
\begin{equation}\label{eq:conditionB}
\text{for all $a, b \ge 0$ with $a + b < m$, the coefficient of $T^aY^b$ in $Q(\alpha_i+T, y_i+Y)$ must be zero.}
\end{equation}
We can prove that (\ref{eq:conditionB}) implies the target (\ref{eq:conditionE}) when $d=0$ so it suffices to ensure that (\ref{eq:conditionB}) holds for all $i\in[n]$. Indeed these coefficients can be expressed as a known linear combination of $(c_f)_{f\in\mathcal{Q}}$. For each $i\in[n]$, (\ref{eq:conditionB}) introduces $\binom{m+1}{2}$ linear constraints, so there are $M=\binom{m+1}{2}n$ constraints in total. To ensure such a solution exists, it suffices to set parameters so that $|\mathcal{Q}|>M$. A suitable calculation shows this is only possible when $A>\sqrt{nk}$. Importantly, using (\ref{eq:conditionB}) as a proxy for (\ref{eq:conditionE}) incurs a substantial inefficiency in the derivation of $Q$. As such, the Guruswami--Sudan algorithm is unable to go beyond the Johnson bound.

\subsection*{Acknowledgment and Statement on AI Usage}
Zihan Zhang and Kai Zhe Zheng thank the Simons Institute for the Theory of
Computing for its hospitality and support. This work was carried out
during their Research Fellowships at the Simons Institute as part of the
program on Pseudorandomness and High-Dimensional Expansion. Kai Zhe Zheng is also grateful to Scott Duke Kominers and Justin Thaler for introducing him to \url{better.codes}.

The authors were initially inspired by a submission to the crowd-sourced
Proximity Prize effort on \url{better.codes}, where a proof was published by user \texttt{nasqret}
that a specific blocklength $2^{18}$, rate $1/2$ Reed--Solomon code established a combinatorial list size bound up to radius
\[
\frac{76790}{2^{18}} \approx 0.2929306,
\]
slightly beyond the Johnson radius $1-1/\sqrt{2}$. This submission was, to the best of our knowledge, the
first meaningful step beyond the Johnson radius for a specific evaluation
domain, and a similar form of the $d=1$ version of the techniques in this
paper appeared there.\footnote{See the submission here: \url{https://github.com/proximity-prize/proximity-prize/pull/122}.} After processing the techniques therein, the authors aimed to uncover the core mathematical novelty responsible for the improvement and 
subsequently extended the techniques to the general setting developed here. They relied on AI interaction, specifically GPT-5.6 Sol, to help produce the rest of the results of this paper. All mathematical statements, proofs, and conclusions were subsequently
developed and independently verified by the authors, who take full responsibility for the contents of this paper.

The author are grateful to Omar Alrabiah, Rohan Goyal, and Venkatesan Guruswami for pointing out to us, after the first version of this paper was posted, the simple reduction that extends our low-rate result to arbitrary constant rates.

\section{Preliminaries}

Unless otherwise stated, we let $q$ denote a sufficiently large prime.

\paragraph{Hasse derivatives.} Given a univariate polynomial $P = \sum_{i=0}^{k-1} a_i X^i \in \F_q[X]$, we define its $\ell$-th Hasse derivative (e.g., \cite{DKSS13}) to be
\[
P^{[\ell]}(X) := \sum_{i=\ell}^{k-1} \binom{i}{\ell} a_i X^{i-\ell}. 
\]
Crucially, the Hasse derivatives satisfy the identity that
\begin{align}
    P(X + T) = \sum_{\ell=0}^{k-1} P^{[\ell]}(X) T^{\ell}.\label{eq:forward}
\end{align}
Substituting $X' = X+T$ and $T' = -T$ into the above formula gives us the M\"obius inversion
\begin{align}
    P(X) = \sum_{\ell=0}^{k-1} P^{[\ell]}(X+T) (-T)^{\ell},\label{eq:backward}
\end{align}
which we make key use of in our interpolation.

\paragraph{Weighted Degrees.}

Given a multivariate polynomial $Q \in \F_q[X_1, \hdots, X_d]$, we define the weighted degree of $Q$ with respect to weights $w_1, \hdots, w_d \in \mathbb Z_{\ge 0}$ to be the maximum value of $w_1i_1 + \cdots + w_di_d$ among all monomials $X_1^{i_1} \cdots X_d^{i_d}$ appearing in $Q$. We denote this quantity by $\deg_{w_1, \hdots, w_d} Q$. In the special case where $w_i = 1$ and $w_j = 0$ for all $j \in [d] \setminus \{i\}$, we denote this weighted degree more succinctly by $\deg_{X_i} Q$.

\paragraph{Solving Polynomial Differential Equations.}

A crucial ingredient in our list-decoding algorithm is finding all univariate polynomials $P \in \F_q[X]$ with $\deg P \le k-1$ which satisfy the identity
\begin{align}
    Q(X, P(X), P^{[1]}(X), \hdots, P^{[d]}(X)) \equiv 0\label{eq:Q-root-condition}
\end{align}
for some (fixed) multivariable polynomial $Q \in \F_q[X, Y_0, Y_1, \hdots, Y_d]$. By a result of Kopparty~\cite{Kopparty2015} (see also \cite{KT22}), as long as our prime field $q$ is larger than some suitable weighted degrees of $Q$, we can find a list of all such $P$ efficiently. We state this result as follows.

\begin{theorem}[Theorem~4.3~\cite{Kopparty2015}, restated]\label{thm:root-finding}
Assume $q \ge k > d$ with $q$ prime. Assume nonzero $Q \in \F_q[X, Y_0, \hdots, Y_d]$ satisfies $\deg_{Y_i} Q < q$ for all $i \in \{0,1,\hdots, d\}$ and that
\[
    \deg_{1,k-1,k-2, \hdots, k-d-1} Q < q^2,
\]
then one can find, in $q^{O(d+1)}$ time, all $P \in \F_q[X]$ of degree at most $k-1$ satisfying (\ref{eq:Q-root-condition}). Furthermore, the number of such $P$ found is at most $q^{4d+6}$.
\end{theorem}

\section{Novel Interpolation via Hidden Derivatives}\label{sec:interpolation}

Our goal in this section is to interpolate an \emph{explainer polynomial} $Q\in \mathbb F_q[X,Y_0,Y_1,\ldots,Y_d]$ which at a high level has the following two properties
\begin{itemize}
    \item $Q$ is ``low-degree'' under some suitable notions of low-degree.
    \item For any degree at most $k-1$ polynomial $P$ with desired agreement with the received word $\vec{y}=(y_1,\dots,y_n)\in\F_q^n$, $P$ satisfies
    \[
    Q(X, P(X), P^{[1]}(X), \ldots, P^{[d]}(X)) \equiv 0 .
    \]
\end{itemize}
By computing $Q$ satisfying these two properties, we can later efficiently find all low-degree polynomials sufficiently correlated with the received word by finding all $P$ satisfying the relation in the second item. 

\paragraph{List of Parameters.}
Fix a target list decoding radius $1-\eps$ and a slack $\theta \in (0,1)$ as in \Cref{thm:main}. We set 
\begin{equation} \label{eq:global_params}
A = \left\lceil
\eps n
\right\rceil, \qquad
d=\left\lceil \eps^{-3/\theta} \right \rceil,
\qquad
m=d^3,
\end{equation}
\begin{equation}\label{eq:s-and-B}
\qquad
B=\left\lceil\frac{mA}{k-1}\right\rceil.
\end{equation}
We additionally make the rate assumption that

\begin{align}
    R = \frac{k}{n} \le (1-\theta)\eps\label{eq:rate-bound}
\end{align}

Intuitively, $A$ represents the minimum agreement for our decoder, $d$ is the maximum order of a Hasse derivative considered by $Q$, $m$ is the multiplicity we impose on each received point, and $B$ is the maximum degree of $Q$ with respect to the inputs $Y_0, \hdots, Y_d$.

\paragraph{Local Constraints at Each Received Point $(\alpha,y)$.}
We require $Q$ to satisfy a series of local constraints at each received point $(\alpha,y)$.
The purpose of these constraints is to ensure that if a candidate polynomial $P$ satisfies
$P(\alpha)=y$, then the specialization
\[
Q\bigl(X,P(X),P^{[1]}(X),\ldots,P^{[d]}(X)\bigr)
\]
has a zero of multiplicity at least $m$ at $X=\alpha$. The constraint is motivated by trying to match the low degree terms in \eqref{eq: backward taylor}. In particular, setting $P(\alpha)=y$ there and looking at terms of $T$-degree at most $d$, we can write
\begin{equation} \label{eq: backward taylor}
P(\alpha+T)
\equiv
y
+TP^{[1]}(\alpha+T)
-T^2P^{[2]}(\alpha+T)
+\cdots
+(-1)^{d+1}T^dP^{[d]}(\alpha+T)
\mod{T^{d+1}}.
\end{equation}
Thus, although we do not know the derivatives of $P$ at the time of
interpolation, any polynomial passing through $(\alpha,y)$, along with its derivatives, must satisfy this formal relation.

We encode this relation by introducing formal variables
$Y_1,\ldots,Y_d$ and a remainder variable $E$, and making the substitution
\[
X=\alpha+T,
\qquad
Y_0
=
y+TY_1-T^2Y_2+\cdots+(-1)^{d+1}T^dY_d+TE.
\]
For a polynomial $P$ satisfying $P(\alpha)=y$, after setting $Y_j=P^{[j]}(\alpha+T)$,
the corresponding remainder satisfies
\[
E\equiv0\mod{T^d}.
\]
Hence, our explainer polynomial can encode this relation as follows. Make the substitution
\begin{equation}\label{eq:local-Q-expand}
  \begin{split}
  Q^{(\alpha,y)}(T,E,Y_1,\ldots,Y_d)
  &:=
  Q\left(\alpha+T,
    y+\sum_{j=1}^{d}(-1)^{j+1}T^jY_j+TE,
    Y_1,\ldots,Y_d\right) \\
  &=\sum_{b,\vec{e}}q_{b,\vec{e}}(T)E^bY_1^{e_1}\cdots Y_d^{e_d},
  \end{split}
\end{equation}
where the $q_{b,\vec{e}}(T) \in \F_q[T]$ are coefficients. We enforce the following local constraint for every $(\alpha,y)$ and nonnegative $(b,\vec{e}):=(b,e_1,e_2,\dots,e_d) \in \mathbb{Z}^{d+1}$ such that $db<m$:
\begin{equation}\label{eq:local-condition}
  q_{b,\vec{e}}(T)\equiv0\mod{T^{m-db}}
\end{equation}
Note that these are homogeneous linear constraints on the coefficients of \(Q\).

\begin{lemma}\label{lem:contact}
Suppose \(P(\alpha)=y\), and suppose \(Q\) satisfies
\eqref{eq:local-condition} at \((\alpha,y)\).  Then
\[
Q\!\left(
\alpha+T,\,
P(\alpha+T),\,
P^{[1]}(\alpha+T),\ldots,P^{[d]}(\alpha+T)
\right)
\equiv0\mod{T^m}.
\]
\end{lemma}
\begin{proof}
We can write
\[
Q\!\left(
\alpha+T,\,
y+\sum_{j=1}^d(-1)^{j+1}T^jY_j+TE,\,
Y_1,\ldots,Y_d
\right)
=
\sum_{b,\vec{e}}q_{b,\vec{e}}(T)E^bY_1^{e_1}\cdots Y_d^{e_d}.
\]
Now substitute
\[
Y_j=P^{[j]}(\alpha+T)
\]
for each $j \in [d]$ and
\begin{equation} \label{eq: tmp expansion}
E=E_P(T):=
\frac{
P(\alpha+T)-y-\sum_{j=1}^d(-1)^{j+1}T^jP^{[j]}(\alpha+T)
}{T}.
\end{equation}
Under this substitution, the second input to $Q$ now becomes $P(\alpha+T)$ and by \eqref{eq: backward taylor}, we have $E_P(T)\equiv0\mod{T^d}.$
Hence $E_P(T)^b$ is divisible by $T^{db}$.  If $db<m$, the local
constraint imposed by \eqref{eq:local-condition} gives
\[
q_{b,\vec{e}}(T)\equiv0\mod{T^{m-db}},
\]
so $q_{b,\vec{e}}(T)E_P(T)^b$ is divisible by $T^m$.  If $db\ge m$,
then $E_P(T)^b$ is already divisible by $T^m$.  Thus every summand in \eqref{eq:local-Q-expand} is 
divisible by $T^m$, and therefore
\[
Q\!\left(
\alpha+T,\,
P(\alpha+T),\,
P^{[1]}(\alpha+T),\ldots,P^{[d]}(\alpha+T)
\right)
\equiv0\mod{T^m}.\qedhere
\]
\end{proof}

\subsection{Dimension of the Interpolation Space}

Here we define the monomials that we use when interpolating $Q$ and give
a lower bound on their number. For
\[
c=(c_2,\ldots,c_{d})
\in\mathbb Z_{\ge0}^{d-1},
\]
it will be helpful to define.
\[
\omega(c)
=
\sum_{j=2}^{d}(j-1)c_j,
\qquad
|c|
=
\sum_{j=2}^{d}c_j,
\qquad
Y^c
=
\prod_{j=2}^{d}Y_j^{c_j}.
\]
The function $\omega$ tracks a weighted-degree like quantity, assigning weight $j-1$ to the variable $Y_j$. To give some intuition, it will be used as follows. When writing out a Taylor expansion and replacing successive derivatives with $Y_1, \ldots, Y_d$, we get an expression that looks like $\sum_{j \geq 0} (-1)^{j} T^{j-1} Y_j$. Then, $\omega(c)$ gives the degree of $T$ in any coefficient of the monomial $Y^c$ when raising this expression to any power. To see this, observe that $Y_j$ can only occur with $T^{j-1}$, hence a power $Y_j^{c_j}$ must be accompanied by a factor $T^{(j-1)c_j}$. On the other hand, $|c|$, is simply the unweighted degree of the monomial $Y^c$.

Now we set
\begin{equation}\label{eq:W-def}
    W := \left\lfloor \frac{(1+\theta/2) dm}{\log(ed)} \right \rfloor 
\end{equation}

Let
\[
\mathcal B
=
\left\{
c\in\mathbb Z_{\ge0}^{d-1}:
\omega(c)\le W,\quad
|c|\le \left \lceil \left(1 + \frac{3\theta}{4} \right) m \right \rceil
\right\}.
\]
and finally, let our space of interpolants  $\mathcal Q
\subseteq
\mathbb F_q[X,Y_0,\ldots,Y_{d}]$
be the span of the monomials
\begin{equation}\label{eq:global-space}
\left\{
X^aY_0^{b_0}Y_1^{b_1}Y^c:
\begin{array}{l}
a,b_0\in\mathbb Z_{\ge0},\quad
c\in\mathcal B,\quad
0\le b_1\le m,\\
b_0+b_1+|c|\le B,\\
a+(k-1)(b_0+b_1+|c|)<mA
\end{array}
\right\}.
\end{equation}

At a high level, requiring $c \in \mathcal{B}$ restricts both ordinary degree and $\omega$-weighted degree of monomials in the variable $Y_2, \ldots, Y_d$, while a separate degree requirement is imposed on the variable $Y_1$ for our interpolant polynomials. 

\begin{lemma}\label{lem:global-budgets}
The interpolation space $\mathcal Q$ satisfies the following properties:
\begin{itemize}
    \item Every $Q\in\mathcal Q$ satisfies
    \[
    \deg_{Y_i}Q
    \le B,
    \]
    for each $i \in \{0,\ldots, d\}$.
    \item For every $Q\in\mathcal Q$ and every
    $P\in\mathbb F_q[X]$ with $\deg P\le k-1$,
    \[
    \deg_X
    Q\bigl(X,P(X),P^{[1]}(X),\ldots,P^{[d]}(X)\bigr)
    <mA.
    \]

    \item The dimension of $\mathcal Q$ satisfies
    \[
    \dim\mathcal Q
    >
    \frac{\theta^3}{384}
    |\mathcal B|(k-1)m^3.
    \]
\end{itemize}
\end{lemma}
\begin{proof}
The first item follows directly from the definition of $\mathcal Q$.
Indeed, every monomial
\[
X^aY_0^{b_0}Y_1^{b_1}Y^c
\]
in $\mathcal Q$ satisfies
\[
b_0+b_1+|c|\le B.
\]
Hence the exponent of each variable $Y_i$, for $i\in\{0,\ldots,d\}$, is at
most $B$, and therefore
\[
\deg_{Y_i}Q\le B
\qquad
\]
for every $i \in \{0,\ldots, d\}$.

The second item follows straightforwardly too. Fix $P\in\mathbb F_q[X]$ of degree at most $k-1$. Since $\deg P^{[j]}\le k-1$
for every $j$, one can check that
\[
X^a
P(X)^{b_0}
\bigl(P^{[1]}(X)\bigr)^{b_1}
\prod_{j=2}^{d}
\bigl(P^{[j]}(X)\bigr)^{c_j},
\]
has $X$-degree at most
\[
a+(k-1)(b_0+b_1+|c|)<mA.
\]
Hence, any $Q \in \mathcal{Q}$ has $\deg_X
Q\bigl(X,P,P^{[1]},\ldots,P^{[d]}\bigr)
<mA$.

It remains to lower bound $\dim\mathcal Q$. First, by the definition of $A$ in \eqref{eq:global_params} and the fact that $k \leq (1-\theta) \eps$, we have
\[
\frac{A}{k-1}
\ge
\frac{\epsilon n}{(1-\theta)\epsilon n}
=
\frac{1}{1-\theta}
\ge
1+\theta.
\]
On the other hand, for every $c\in\mathcal B$, the definition of
$\mathcal B$ gives
\[
|c|
\le
\left\lfloor
\left(1+\frac{3\theta}{4}\right)m
\right\rfloor
\le
\left(1+\frac{3\theta}{4}\right)m.
\]
Consequently,
\begin{equation}\label{eq:rank-lb-tmp}
m\frac{A}{k-1}-|c|
\ge
\frac{\theta m}{4}.
\end{equation}

Now fix $c\in\mathcal B$ and an integer
\[
0\le b_1\le
\left\lfloor\frac{\theta m}{4}\right\rfloor.
\]
For every pair $(a,b_0)\in\mathbb Z_{\ge0}^2$ satisfying
\[
a+(k-1)b_0
<
mA-(k-1)(|c|+b_1),
\]
observe that the corresponding monomial $X^aY_0^{b_0}Y_1^{b_1}Y^c$
belongs to $\mathcal Q$. Indeed, the weighted-degree condition holds by
construction, and since $a\ge0$,
\[
b_0+b_1+|c|
<
\frac{mA}{k-1}
\le
\left\lceil\frac{mA}{k-1}\right\rceil
=
B.
\]

The number of such pairs $(a,b_0)$ is at least
\[
\frac{
\bigl(mA-(k-1)(|c|+b_1)\bigr)^2
}{
2(k-1)
}.
\]
Using \eqref{eq:rank-lb-tmp}, this is at least
\[
\frac{k-1}{2}
\left(
\frac{\theta m}{4}-b_1
\right)^2.
\]
Therefore,
\[
\begin{aligned}
\dim\mathcal Q
&\ge
\frac{k-1}{2}
|\mathcal B|
\sum_{b_1=0}^{\lfloor\theta m/4\rfloor}
\left(
\frac{\theta m}{4}-b_1
\right)^2\\
&\ge
\frac{k-1}{2}
|\mathcal B|
\int_0^{\theta m/4}
\left(
\frac{\theta m}{4}-x
\right)^2\,dx\\
&=
\frac{\theta^3}{384}
|\mathcal B|(k-1)m^3.
\end{aligned}
\]
This proves the final item.
\end{proof}

\subsection{Upper Bounding The Number of Constraints}

At every received point $(\alpha,y)$, we impose the divisibility
constraints \eqref{eq:local-condition}. It remains to bound how
many independent linear conditions these constraints impose.
For this rank calculation only, it is convenient to write
\[
Y_0=y+TU.
\]
For $z\ge0$, define
\[
\Lambda(z)
=
\left \vert \{c\in\mathbb Z_{\ge0}^{d-1}:\omega(c)\le z  \} \right\vert.
\]

\subsubsection{A Bound on $\Lambda(z)$'s Growth}


As a first, modular step, in this section we will prove the following bound on how quickly $\Lambda(z)$ grows. 

\begin{lemma}
\label{lem:scaled-shell}
We have
\[
\Lambda(z)
\le
\frac{
\left(z+\binom{d}{2}\right)^{d-1}
}{
((d-1)!)^2.
},
\]
Further, if $d \ge \lceil \eps^{-3/\theta}\rceil$ and
\begin{align}
\eps < \left(\frac{\theta^3(1-\theta)}{768}\right)^{\frac{5+\theta}{1-\theta}}\label{eq:eps-theta-tradeoff}
\end{align}
then
\[
\Lambda(W+m) \leq |\mathcal{B}| d^{2/(2+\theta) + \theta(1-\theta)/18}.
\]
\end{lemma}

Before proving this, we require some elementary facts which govern the volume of these scaled lattices. We start by recalling a folklore, basic fact about the volume of the unit simplex:
\begin{proposition}\label{prop:simplexVolume}
    Consider the $d-1$ dimensional simplex 
    \[
    \mathcal{S} = \left \{x \in \mathbb{R}_{\geq 0}^{d-1}: \sum_i x_i \leq 1 \right \}.
    \]
    Then, $\mathrm{Vol}(\mathcal{S}) = \frac{1}{(d-1)!}$.
\end{proposition}

With this bound on the volume of the unit simplex, we can easily extend this to a formula for the volume of the ``scaled'' simplex:

\begin{corollary}\label{cor:scaledVolume}
    Let $z \in \mathbb{R}_{\geq 0}$. Then, for the scaled simplex 
    \[
    \mathcal{S}_z = \left \{ x \in \mathbb{R}_{\geq 0}^{d-1}: \sum_i i \cdot x_i  \leq z\right \},
    \]
    we have that 
    \[
    \mathrm{Vol}(\mathcal{S}_z) = \frac{z^{d-1}}{((d-1)!)^2}.
    \]
\end{corollary}

\begin{proof}
    We start by assuming that $z = 1$. In this case, we consider the body 
    \[
    \left \{ x \in \mathbb{R}_{\geq 0}^{d-1}: \sum_i i \cdot x_i  \leq 1\right \}.
    \]
    Now, we perform a change of variables with $y_i = i \cdot x_i$. In the $y$-variables, the volume of the body 
    \[
    \left \{ y \in \mathbb{R}_{\geq 0}^{d-1}: \sum_i y_i  \leq 1\right \}
    \]
    is exactly $\frac{1}{(d-1)!}$ by \cref{prop:simplexVolume}. Translating back to the $x$ variables, we lose a factor of exactly $(d-1)!$ in the volume. This is because a linear map multiplies the volume by a factor of exactly its determinant (which in this case is exactly $(d-1)!$). Thus, $\mathrm{Vol}(\mathcal{S}_1) = \frac{1}{((d-1)!)^2}$, and uniformly scaling the coordinates by $z$, and thus the volume by a factor of $z^{d-1}$ yields the desired corollary. 
\end{proof}

Our goal is to use these simplices to provide a tighter bound on the size of $\Lambda(z)$. To do this, we now compare the volume of $\mathcal{S}_{z}$ with $\Lambda(z)$. We consider any $u \in \mathbb{Z}^{d-1}_{\geq 0}$ which satisfies $\sum_i i \cdot u_i \leq z$. For any such point, we then consider its corresponding ``half-open unit cube,'' obtained as 
\[
u + [0,1)^{d-1}.
\]
We have:

\begin{proposition}
    Consider any  $u \in \mathbb{Z}^{d-1}_{\geq 0}$ which satisfies $\sum_i i \cdot u_i \leq z$. Then, the unit cubes $u + [0,1)^{d-1}$ are pairwise disjoint and moreover 
    \[
    \mathrm{Vol} \left ( \bigcup_{u \in \mathbb{Z}^{d-1}_{\geq 0}: \sum_i i \cdot u_i \leq z}u + [0,1)^{d-1} \right ) = \Lambda(z). 
    \]
\end{proposition}

\begin{proof}
    Pairwise disjointness follows from the fact that each unit cube is anchored at a distinct integral lattice point. The volume follows from the fact that 
    \[
    \mathrm{Vol} \left ( \bigcup_{u \in \mathbb{Z}^{d-1}_{\geq 0}: \sum_i i \cdot u_i \leq z}u + [0,1)^{d-1} \right ) = \sum_{u \in \mathbb{Z}^{d-1}_{\geq 0}: \sum_i i \cdot u_i \leq z} 1 = \Lambda(z).
    \]
\end{proof}

We now prove the following containment property, which will allow us to directly translate our volume estimates for the scaled simplex to bounds on $\Lambda(z)$:
\begin{proposition}
For $z$ an integer, we have that 
\[
    \mathcal S_z
    \subseteq
    \bigcup_{\sum_i i u_i\le z}\bigl(u+[0,1)^{d-1}\bigr) \subseteq \mathcal{S}_{z + d(d-1)/2}.
    \]
\end{proposition}

\begin{proof}
    The first containment follows because if $x \in \mathcal S_z$, then $u_i = \lfloor x_i \rfloor$ satisfies 
    \[
    \sum_{i = 1}^{d-1} i \cdot u_i  \leq \sum_{i = 1}^{d-1} i \cdot x_i \leq z,
    \]
    and $x \in u+[0,1)^{d-1}$. The second containment follows because for any  $x\in u+[0,1)^{d-1}$ and $\sum_i i u_i\le z$, then
\[
    \sum_{i=1}^{d-1}i x_i
    <
    \sum_{i=1}^{d-1}i(u_i+1)
    \le z+ \sum_{i = 1}^{d-1} i \leq z + d(d-1)/2.
\]
Thus, any such $u+[0,1)^{d-1}$ is contained in $\mathcal{S}_{z + d(d-1)/2}$.
\end{proof}

By using \cref{cor:scaledVolume}, we then see that:

\begin{corollary}\label{eq:lattice-point-sandwich}
     Let $z \in \mathbb{R}_{\geq 0}$. Then, 
     \[
    \frac{z^{d-1}}{((d-1)!)^{2}}
    \le
    \Lambda(z)
    \le
    \frac{(z+d(d-1)/2)^{d-1}}{((d-1)!)^{2}}.\]
\end{corollary}

This estimate of $\Lambda(z)$ already proves the first point of \cref{lem:scaled-shell}. 

Before the proof of the second item of \cref{lem:scaled-shell}, we introduce one key lemma that we will take advantage of, which relates the expectation of a distribution over a simplex to its centroid: 

\begin{lemma}\label{lem:simplexExpectation}
Let $S=\operatorname{conv}\{v_0,\ldots,v_{d-1}\}$
be a $d-1$-dimensional simplex, and let $X$ be uniformly distributed on $S$. Then
\[
\mathbb{E}[X]=\frac{1}{d}\sum_{i=0}^{d-1} v_i.
\]
In particular, the expectation of $X$ is the centroid of $S$.
\end{lemma}

We provide a self-contained proof of this lemma in the appendix, as it is a folklore result with a simple proof.

\begin{proof}[Proof of the second item of \cref{lem:scaled-shell}.]
Recall that we wish to show that $\Lambda(W+m) \leq |\mathcal{B}| d^{2/(2+\theta) + \theta(1-\theta)/18}$. $\mathcal{B}$ is already defined with respect to $W$, indeed, 
\[
\mathcal B
=
\left\{
c\in\mathbb Z_{\ge0}^{d-1}:
\omega(c)\le W,\quad
|c|\le \left \lceil \left(1 + \frac{3\theta}{4} \right) m \right \rceil
\right\}.
\]
If we remove the second constraint from the definition of $\mathcal{B}$ (i.e., the restriction on $|c|$), our desired inequality is not hard to show, as it boils down to comparing the number of points with $\omega(c) \leq W$ to the number of points with $\omega(c) \leq W + m$. This could be directly bounded by (\ref{eq:lattice-point-sandwich}). Unfortunately, $\mathcal{B}$ is also defined with the second constraint on $|c|$. Our goal going forward is thus to show that this constraint \emph{does not} decrease the number of points in $\mathcal{B}$ by too much. 

Next, recall that $W := \left\lfloor \frac{(1+\theta/2) dm}{\log(ed)} \right \rfloor $, that $m = d^3$, and that $d := \lceil \eps^{-3/\theta}\rceil $. Since $e^x \ge x+1$ for all $x \in \mathbb R$, we have that $d \ge \log(d) + 1 = \log(ed)$. Therefore, $W \ge \left\lfloor (1+\theta/2)m\right\rfloor \ge 1$. 
Now, our intermediate goal is to show that \emph{most points} in $\mathcal{S}_W$ also satisfy the additional constraint $|c|\le \left \lceil \left(1 + \frac{3\theta}{4} \right) m\right \rceil$. To do this, we sample a random point $X = (X_1, \dots X_{d-1})$ from the simplex $\mathcal{S}_W$, and will bound its expected magnitude. Crucially, the simplex $\mathcal{S}_W$ has as its vertices 
\[
0, \quad W \cdot e_1, \quad \frac{W}{2} \cdot e_2, \quad \dots, \quad \frac{W}{d-1} \cdot e_{d-1},
\]
and the expectation of a uniform point in a simplex is its centroid (as per \cref{lem:simplexExpectation}). Thus, we see that 
\[
\E \left [ \sum_{i = 1}^{d-1} X_i\right ] = \frac{W}{d} \cdot \sum_{i = 1}^{d-1} \frac{1}{i} \leq \frac{W \log(ed)}{d} \leq \left ( 1 + \theta/2\right ) \cdot m,
\]
by our choice of $W$. Now, we can apply a simple Markov bound to see that 
\[
\Pr \left [\sum_{i = 1}^{d-1} X_i > \left \lceil \left(1 + \frac{3\theta}{4} \right) m\right \rceil\right ] \leq \frac{\left ( 1 + \theta/2\right )  }{\left(1 + 3\theta/4 \right)}.
\]

Equivalently, this means that 
\[
\mathrm{Vol} \left ( \left \{X \in \mathcal{S}_W: |X| \leq  \left \lceil \left(1 + \frac{3\theta}{4} \right) m\right \rceil \right \}\right ) \geq \left ( 1 - \frac{\left ( 1 + \theta/2\right )  }{\left(1 + 3\theta/4 \right)}\right ) \cdot \mathrm{Vol}(\mathcal{S}_W) = \frac{\theta}{4 + 3\theta} \cdot \frac{W^{d-1}}{((d-1)!)^2},
\]
where the last equality uses \cref{cor:scaledVolume}.

Now, we claim that 
\[
\left \{X \in \mathcal{S}_W: |X| \leq  \left \lceil \left(1 + \frac{3\theta}{4} \right) m\right \rceil \right \} \subseteq \bigcup_{u \in \mathcal{B}} \left ( u + [0,1)^{d-1} \right ).
\]
This is because if we consider any point $x$ in the first set, if we apply the coordinate-wise floor $\lfloor x_i \rfloor$, the resulting point $u$ now must be in the set $\mathcal{B}$. Adding the half-open unit cube to every point $u \in \mathcal{B}$ thus captures all points whose coordinate-wise floor would be in $\mathcal{B}$. Importantly, by plugging in our volume bounds, we know that 
\[
|\mathcal{B}| = \mathrm{Vol} \left ( \bigcup_{u \in \mathcal{B}} \left ( u + [0,1)^{d-1} \right )\right ) \geq  \mathrm{Vol} \left ( \left \{X \in \mathcal{S}_W: |X| \leq  \left \lceil \left(1 + \frac{3\theta}{4} \right) m\right \rceil \right \}\right ) 
\]
\[
\geq \frac{\theta}{4 + 3\theta} \cdot \frac{W^{d-1}}{((d-1)!)^2}.
\]

At the same time, (\ref{eq:lattice-point-sandwich}) implies that 
\[
\Lambda(W + m) \leq \frac{\left (W + m + \binom{d}{2} \right )^{d-1}}{((d-1)!)^2}.
\]
Dividing these two estimates then implies that 
\[
\frac{\Lambda(W+m)}{|\mathcal{B}|} \leq \frac{4 + 3 \theta}{\theta} \cdot \left ( 1 + \frac{m + \binom{d}{2}}{W}\right )^{d-1}.
\]
This is exactly the quantity we seek to bound for \cref{lem:scaled-shell}, and we now simply plug in the relationships between $m, d, W, \theta$ to obtain our desired result. 

 Since $\log(1+x)\leq x$, $m=d^3$, and
\[
    W
    =
    \left\lfloor
        \frac{(1+\theta/2)d^4}{\log(ed)}
    \right\rfloor,
\]
we have
\begin{align*}
    \log\frac{\Lambda(W+m)}{|\mathcal B|}
    &\leq
    \log\frac{4+3\theta}{\theta}
    +
    \frac{(d-1)\left(m+\binom{d}{2}\right)}{W} \\
    &=
    \frac{2}{2+\theta}\log d+O_\theta(1).
\end{align*}
Here the final equality also absorbs the floor in the definition of $W$. Since $d = \lceil \eps^{-3/\theta}\rceil$ and (\ref{eq:eps-theta-tradeoff}) holds, we have that the last expression is at most
\[
    \left(
        \frac{2}{2+\theta}
        +
        \frac{\theta(1-\theta)}{18}
    \right)
    \log d.
\]
Exponentiating gives
\[
    \Lambda(W+m)
    \leq
    |\mathcal B|d^{\,2/(2+\theta)+\theta(1-\theta)/18},
\]
as desired.
\end{proof}

\subsubsection{Bounding the Rank}
With this bound on the growth of $\Lambda(z)$, we can now proceed to bound the rank of all the imposed constraints. Throughout this section, we fix one pair $(\alpha, y)$ where $\alpha$ is an evaluation point and $y$ is a received value. We will bound the number of linearly independent constraints imposed across all of the coefficient constraints arising from \eqref{eq:local-condition} applied at $(\alpha, y)$. 

To this end, we will think of the constraints as a linear map from the space of interpolation monomials
$\mathcal{Q}$, into some suitable $\F_q$ vector space. Specifically, let 
\[
\mathcal I
=
\left\{
(i, b, e_1,c) \in \mathbb{Z}_{\geq 0}^{d+2}:
\begin{array}{l}
i,b,e_1\in\mathbb Z_{\ge0},\quad
c\in\mathbb Z_{\ge0}^{d-1},\\
i+db<m,\quad
e_1\le2m-1,\quad
\omega(c)\le W+m-1
\end{array}
\right\},
\]
denote the set of possible exponent vectors indexing the monomials over $T, E, Y_1,\ldots, Y_d$ that are required to vanish 
by \eqref{eq:local-condition}. Then, we let $\Phi_{\alpha,y}: \mathcal{Q} \to \F_q^{\mathcal{I}}$, where $\Phi_{\alpha,y}(Q)_{(i,b,e_1,c)}$ is the coefficient of $T^iE^bY_1^{e_1}Y^c$ in $Q\left(
\alpha+T,\,
y+\sum_{j=1}^d(-1)^{j+1}T^jY_j+TE,\,
Y_1,\ldots,Y_d
\right)$, i.e.\ the substitution in \eqref{eq:local-Q-expand} when defining the constraints in \eqref{eq:local-condition}. In this language, the constraints of \eqref{eq:local-condition} imposed by $(\alpha, y)$ are exactly equivalent to requiring $\Phi_{\alpha, y}$ to vanish. Hence, the goal of this subsection is  to bound the rank of $\Phi_{\alpha,y}$. Multiplying this bound by the number of received points then gives a bound on the overall number of linearly independent constraints \eqref{eq:local-condition} we impose.

To bound $\rank(\Phi)$, it will be helpful to first perform one intermediate substitution towards \eqref{eq:local-Q-expand}, and instead bound the rank of the map \emph{after} this substitution. Let $\Psi_{\alpha,y}$ be the map which takes $Q \in \mathcal{Q}$ as input, substitutes
\[
X=\alpha+T,
\qquad
Y_0=y+TU,
\]
and reduces modulo $T^m$. Notice that every
monomial appearing lies in
\[
\mathcal V
=
\operatorname{span}
\left\{
T^rU^aY_1^bY^c:
0\le r<m,\;
0\le a\le r,\;
0\le b\le m,\;
\omega(c)\le W+r
\right\}.
\]
Indeed, every factor of $U$ arising from $(y+TU)^{b_0}$ is accompanied
by a factor of $T$, while the conditions $b_1\le m$ and $\omega(c)\le W$ already hold
for every monomial of $\mathcal Q$. Hence $\Psi_{\alpha,y}$ maps from $\mathcal{Q} \to \mathcal{V}$. 

By performing this substitution and modular reduction, we can instead focus
on bounding the ``remaining'' transformation needed for
$\Phi_{\alpha,y}$, which we define by $\Gamma:\mathcal V\to \F_q^{\mathcal I}.$
For $F\in\mathcal V$, rewrite $F$ in the variables
$T,E,Y_1,\ldots,Y_d$ using
\begin{equation}\label{eq:E-sub}
E
=
U-\sum_{j=1}^d(-1)^{j+1}T^{j-1}Y_j.
\end{equation}
Then, for every $(i,b,e_1,c)\in\mathcal I$, define
$\Gamma(F)_{(i,b,e_1,c)}$ to be the coefficient of $T^iE^bY_1^{e_1}Y^c$
in the resulting polynomial.
Thus, $\Gamma(F)$ records exactly the coefficients required to vanish by
the local constraints.

One can check that the map of interest $\Phi_{\alpha, y}$ factors as $\Phi_{\alpha,y}
=
\Gamma\circ\Psi_{\alpha,y}.$
Therefore,
\[
\operatorname{rank} \Phi_{\alpha,y}
\le
\operatorname{rank}\Gamma,
\]
and it suffices to upper bound $\operatorname{rank} \Gamma$.

Our strategy is to find a large number of linearly independent vectors in $\ker \Gamma$, and to this end we start by observing a simple divisibility condition which ensures membership in $\ker \Gamma$.
\begin{lemma}\label{lem:kernel-criterion}
Suppose $F\in\mathcal V$ is divisible by $T^rE^h$ for some $r,h\ge0$
satisfying $r+dh\ge m.$ Then $F\in\ker\Gamma$.
\end{lemma}

\begin{proof}
After rewriting $F$ in the variables $T,E,Y_1,\ldots,Y_d$, every monomial
containing $E^b$ has $b\ge h$ and is divisible by $T^r$. Hence
\[
r+db\ge r+dh\ge m.
\]
Hence after rewriting $F$, none of these monomials are recorded, and
$\Gamma(F)=0$.
\end{proof}

For $0\le r<m$, define
\[
h_r:=\left\lceil\frac{m-r}{d}\right\rceil,
\]
so that $h_r$ is the smallest integer satisfying $r+dh_r\ge m$. Using \cref{lem:kernel-criterion}, we can show that $\Gamma$ has a large kernel by finding many multiples of $T^r E^h$ in $\mathcal{V}$, for suitable $r,h$, that are linearly independent.

\begin{lemma}\label{lem:large-kernel}
We have
\begin{equation}\label{eq:kernel-lower-bound}
\dim\ker\Gamma
\ge
\sum_{r=0}^{m-1}
\max(0,(r-h_r+1)(m-h_r+1)\Lambda(W+r)).
\end{equation}
\end{lemma}

\begin{proof}For each $0\le r<m$, define
\[
\mathcal K_r
=
\operatorname{span}
\left\{
T^rE^{h_r}U^aY_1^bY^c:
0\le a\le r-h_r,\;
0\le b\le m-h_r,\;
c\in\mathbb Z_{\ge0}^{d-1},\;
\omega(c)\le W+r
\right\}.
\]
Here, $\mathcal{K}_r$ is a subspace of $\F_q[T,U,Y_1,\ldots,Y_d]/(T^m)$ and we define $\mathcal K_r=0$ if there are no monomials. We first verify that $\mathcal K_r\subseteq\mathcal V$ as this may not be true apriori. Fix a monomial $T^rE^{h_r}g\in\mathcal K_r$. We verify that it indeed lies in $\mathcal{V}$. Expand $E$ using \eqref{eq:E-sub} and consider a monomial appearing. Let
$e_j$ be the exponent of $Y_j$ in it for $j \geq 2$ and set $M=\sum_{j=2}^d(j-1)e_j.$ It suffices to show this monomial is in $\mathcal{V}$. Suppose that, in expanding $E^{h_r}$, we select $e_j$ copies of $Y_j$
for each $j\ge2$, and let
\[
M=\sum_{j=2}^d (j-1)e_j.
\]
These choices contribute a factor of $T^M$ and a monomial in
$Y_2,\ldots,Y_d$ of $\omega$-weight $M$. Since the monomial $Y^c$
coming from $g$ satisfies $\omega(c)\le W+r$, the resulting monomial
has $T$-degree $r+M$ and $\omega$-weight at most $W+r+M$. Its $U$-degree is at most
\[
(r-h_r)+h_r=r\le r+M,
\]
and its $Y_1$-degree is at most
\[
(m-h_r)+h_r=m.
\]
Thus, after reducing modulo $T^m$, every resulting monomial lies in
$\mathcal V$.

Additionally, $r+dh_r\ge m,$
so \cref{lem:kernel-criterion} implies $\mathcal K_r\subseteq\ker\Gamma$ and it is straightforward to compute that the dimension of $\mathcal K_r$ is
\begin{equation}\label{eq:ker-tmp-dim}
\dim\mathcal K_r
=
\max(0, (r-h_r+1)(m-h_r+1)\Lambda(W+r)).
\end{equation}
Finally, we check that the spaces $\mathcal K_0,\ldots,\mathcal K_{m-1}$ are linearly
independent. To this end, suppose
\[
\sum_{r=0}^{m-1}F_r=0,
\qquad
F_r\in\mathcal K_r,
\]
and let $r^\star$ be the smallest index for which
$F_{r^\star}\ne0$. Since
\[
E\equiv U-Y_1\pmod T,
\]
the coefficient of $T^{r^\star}$ in $F_{r^\star}$ is
\[
(U-Y_1)^{h_{r^\star}}g_{r^\star}\ne0.
\]
Every $F_r$ with $r>r^\star$ is divisible by $T^{r^\star+1}$, so none of them can cancel out this coefficient, contradicting
$\sum_rF_r=0$. Hence the spaces $\mathcal K_r$ are linearly independent. Summing \eqref{eq:ker-tmp-dim} proves the lemma.
\end{proof}
We are now ready to upper bound the rank of $\Gamma$.

\begin{lemma}[Local rank]\label{lem:local-rank}
We have
\begin{equation}\label{eq:local-rank}
\operatorname{rank}\Gamma
<
|\mathcal B|m^3d^{-\frac{2\theta}{5+\theta}}.
\end{equation}
\end{lemma}
\begin{proof}
For each fixed $T$-degree $r$, we can directly count that there are
$(r+1)(m+1)\Lambda(W+r)$-many
monomials in $\mathcal V$, which means
\[
\dim\mathcal V
=
\sum_{r=0}^{m-1}
(r+1)(m+1)\Lambda(W+r).
\]
By Lemma~\ref{lem:large-kernel} and rank-nullity,
\[
\operatorname{rank}(\Phi)
\le
\sum_{r=0}^{m-1}
\Lambda(W+r)B_r,
\]
where
\[
B_r
:=
(r+1)(m+1)
-
\max(0,(r-h_r+1)(m-h_r+1)).
\]
Observe that by definition of $h_r$, we have
\[
B_r\le h_r(r+m+2)
\qquad\text{and}\qquad
h_r\le \frac{m-r}{d}+1.
\]
Summing over $0\le r<m$ gives
\begin{align*}
\sum_{r=0}^{m-1}B_r
&\le
\frac{m^3-m}{6d}
+\frac{(m+2)m(m+1)}{2d}
+\frac{m(m-1)}2
+m(m+2) \\
&=
\frac{
d^2(d+1)(2d+1)(d^2-d+1)(2d^2-d+5)
}{6} \\
&\le d^8
=
\frac{m^3}{d},
\end{align*}
where the last inequality holds for $d\ge3$.
Finally, by \cref{lem:scaled-shell},
\[
\Lambda(W+r)
\le
\Lambda(W+m)
\le
|\mathcal B|
d^{\frac{2}{2+\theta}+\frac{\theta(1-\theta)}{18}}.
\]
Moreover,
\[
\frac{2}{2+\theta}
+
\frac{\theta(1-\theta)}{18}
<
\frac{5-\theta}{5+\theta}.
\]
Hence
\[
\Lambda(W+r)
<
|\mathcal B|d^{\frac{5-\theta}{5+\theta}},
\]
and therefore
\[
\operatorname{rank} \Gamma
<
|\mathcal B|d^{\frac{5-\theta}{5+\theta}}
\sum_{r=0}^{m-1}B_r
<
|\mathcal B|m^3d^{-\frac{2\theta}{5+\theta}}.
\]
\end{proof}

\subsection{The Interpolation Lemma}
We now complete the interpolation, which achieves the main goal of this section. 

\begin{proposition}[Interpolation]\label{prop:interpolation}
Let degree $k$ and blocklength $n$ be such that $k/n \le (1-\theta)\eps$, let $A = \left\lceil \eps n \right\rceil$, and let $\mathcal{Q} \subseteq \F_q[X,Y_0,\ldots, Y_d]$ be the interpolation space described in \eqref{eq:global-space}. Assume further that (\ref{eq:eps-theta-tradeoff}) holds.
Then, for any received word $\vec{y}=(y_1,\dots,y_n)\in\F_q^n$, there exists a nonzero $Q \in \mathcal{Q}$ such that any polynomial \(P\) of degree at most \(k-1\) agreeing with
$\vec{y}$ in at least \(A\) positions satisfies
\[
Q\bigl(X,P(X),P^{[1]}(X),\ldots,P^{[d]}(X)\bigr)\equiv0.
\]
\end{proposition}
\begin{proof}
By Lemma~\ref{lem:global-budgets} and the inequality on $A$ in (\ref{eq:global_params}),
\[
\dim\mathcal Q \geq \frac{\theta^3}{384}|\mathcal{B}|(k-1) m^3
\]
On the other hand, Lemma~\ref{lem:local-rank} shows that the number of homogeneous linearly independent constraints imposed over all instances of \eqref{eq:local-condition} is at most $n|\mathcal{B}| m^3 d^{\frac{-2\theta}{5+\theta}}.$
We first verify that $\dim\mathcal Q>n|\mathcal{B}| m^3 d^{\frac{-2\theta}{5+\theta}}$, which will imply that the
linear system has a nonzero solution $Q$ satisfying \eqref{eq:local-condition}. To see this, note that by the inequality above, it is sufficient to show 
\[
\frac{\theta^3}{384} \frac{k-1}{n}d^{\frac{2\theta}{5+\theta}} > 1.
\]
Since we set $k =\left \lfloor (1-\theta) \eps n \right\rfloor$, we have 
\[
\frac{k-1}{n} \geq \frac{(1-\theta)\eps}{2},
\]
and by choice of $d$ in \eqref{eq:global_params}, we have 
\[
d^{\frac{2\theta}{5+\theta}} \geq \eps^{-\frac{6}{5+\theta}}.
\]
Altogether, this gives 
\[
\frac{\theta^3}{384} \frac{k-1}{n}d^{\frac{2\theta}{5+\theta}}  \geq \frac{\theta^3(1-\theta)}{768} \eps^{-\frac{(1-\theta)}{(5+\theta)}},
\]
which is greater than $1$ for all sufficiently small $\eps$---precisely when (\ref{eq:eps-theta-tradeoff}) holds. 

 Fix this $Q$. We show that it satisfies the assertion of the lemma. Let $P \in \F_q[X]$ be a polynomial of degree at most $k-1$. Then, at every agreement evaluation point \(\alpha\) between $P$ and $\vec{y}$,
Lemma~\ref{lem:contact} gives
\[
Q\!\left(
\alpha+T,P(\alpha+T),P^{[1]}(\alpha+T),\ldots,P^{[d]}(\alpha+T)
\right)
\equiv0\mod{T^m}.
\]
Hence the specialization
\[
Q\bigl(X,P(X),P^{[1]}(X),\ldots,P^{[d]}(X)\bigr)
\]
has at least \(A\) distinct roots, each of multiplicity at least \(m\).
It therefore has at least \(mA\) roots counted with multiplicity.
By the second item of Lemma~\ref{lem:global-budgets}, its degree is strictly smaller than
\(mA\), so the polynomial $Q\bigl(X,P(X),P^{[1]}(X),\ldots,P^{[d]}(X)\bigr)$ is identically zero.
\end{proof}



\section{Algorithm and Analysis: Proof of \Cref{thm:main}}\label{sec:algorithm}

The goal of this section is to show how the interpolation techniques of \cref{sec:interpolation} can be leveraged to prove \cref{thm:main}. We begin by presenting our novel list-decoding algorithm.

\subsection{The List-decoding Algorithm}

We present our list-decoding algorithm in Algorithm~\ref{algo:main}. At a high-level, interpolation equations in \cref{sec:interpolation} are used to compute an interpolation polynomial $Q \in \F_q[X, Y_0, Y_1, \hdots, Y_d]$ of suitable degree. Then, the algorithm of Kopparty~\cite{Kopparty2015} is invoked (as stated in \cref{thm:root-finding}) to find the list of potential message polynomials $P$.

\begin{algorithm}[t]\label{algo:main}
    \caption{List-decoding of Reed--Solomon Codes}
    \SetAlgoHangIndent{0pt}
    \SetAlgoLined
    \KwIn{Parameters $n,k,q,\theta,\eps$. Evaluation points $(\alpha_1, \hdots, \alpha_n) \in \F_q^n$, Received message $(y_1, \hdots, y_n) \in \F_q^n$.}
    \KwOut{A list $\mathcal L$ of polynomials $P \in \F_q[X]$ of degree at most $k-1$.}
    {
    Set $A = \lceil \eps n\rceil$, $d = \lceil \eps^{-3/\theta}\rceil$, $m = d^3$,$B = \lceil mA/(k-1)\rceil$, $W = \lfloor (1+\theta/2)dm / \log(ed)\rfloor$.
    
    Define the monomial space $\mathcal Q$ as in (\ref{eq:global-space}). Namely,
    \[
    \mathcal Q = \mathrm{span}\left\{
      X^aY_0^{b_0}Y_1^{b_1}Y^c:
    \begin{array}{l}
        a,b_0,b_1, c_2, \hdots, c_d\in\mathbb Z_{\ge0},\\
        \sum_{j=2}^d c_j \le \lceil (1+3\theta/4)m\rceil,\quad b_1\le m,\\
        b_0+b_1+|c|\le B,\quad \sum_{j=2}^d (j-1)c_j \le W\\
        a+(k-1)(b_0+b_1+|c|)<mA
    \end{array}
  \right\}.
    \]  
    
    Invoke \cref{prop:interpolation} to find some nonzero $Q \in \mathcal Q$ satisfying (\ref{eq:local-condition}) at $(\alpha_i, y_i)$ for all $i \in [n]$. That is, for all $i \in [n]$ enforce that the polynomial
    \begin{align}
    Q\left(\alpha_i + T, y - \sum_{j=1}^d(-T)^j Y_j + T^{d+1}E, Y_1, \hdots, Y_d\right)\label{eq:Q-d+1}
    \end{align}
    is divisible by $T^m$.
    
    Invoke \cref{thm:root-finding} to find a list $\mathcal L$ of polynomials $P \in \F_q[X]$ of degree at most $k-1$ for which
    \[
    Q(X, P(X), P^{[1]}(X), \hdots, P^{[d]}(X)) \equiv 0.
    \]
    
    Remove all $P \in \mathcal L$ for which $(P(\alpha_1), \hdots, P(\alpha_n))$ has agreement less than $A$ with $(y_1, \hdots, y_n)$. 
    
    \Return{$\mathcal L$}
    }
\end{algorithm}

\subsection{Analysis}

We now state our formal list-decoding result. 

\begin{theorem}[Main Result]\label{thm:main-formal}
Fix parameters $n \ge k \ge 1$ and $\eps, \theta \in (0,1)$ such that, $k > \lceil \eps^{-3/\theta}\rceil$, $k / n \le (1-\theta)\eps$, and 
\begin{align}
\eps < \left(\frac{\theta^3(1-\theta)}{768}\right)^{\frac{5+\theta}{1-\theta}}.\label{eq:eps-theta}
\end{align}
Further select a prime $q$ such that $q \ge \max(n, 4\eps^{1-9/\theta} n/k)$. Then, for any distinct evaluation points $\alpha_1, \hdots, \alpha_n \in \F_q$ and any received word $(y_1, \hdots, y_n) \in \F_q^n$, then Algorithm~\ref{algo:main} outputs, in $q^{O(\eps^{-12/\theta})}$ time, a list of all univariate polynomials $P \in \F_q[X]$ of degree less than $k$ such that the agreement between $(P(\alpha_1), \hdots, P(\alpha_n))$ and $(y_1, \hdots y_n)$ is at least $\eps n$. Furthermore, the lenght of this list is at most $q^{O(\eps^{-3/\theta})}$. 
\end{theorem}

Before we prove, \cref{thm:main-formal}, we show how \cref{thm:main-formal} implies \cref{thm:main}.

\begin{proof}[Proof of \cref{thm:main}]
For fixed $\theta \in (0, 1)$ it is clear that (\ref{eq:eps-theta}) holds for all all sufficiently small $\eps > 0$. Invoking \cref{thm:main-formal}, for any prime $q \ge \max(n, 4\eps^{1-9/\theta} n/k) = \Theta_{\theta,\eps}(n)$ (which exists by Bertrand's postulate), we have that any Reed-Solomon code with $n$ distinct evaluation points over $\F_q$ of rate at most $(1-\theta)\eps$ can be decoded up to radius $(1 - \eps)n$ in time $q^{O(\eps^{-12/\theta})}$ with a list size at most $q^{O(\eps^{-9/\theta})}$.
\end{proof}

We now turn to proving \cref{thm:main-formal}.

\begin{proof}[Proof of \cref{thm:main-formal}]
We first prove that Algorithm~\ref{algo:main} is correct. For our parameters $k, n, \eps, \theta$ satisfying $k / n \le (1-\theta)\eps$ and (\ref{eq:eps-theta}), we have by \cref{prop:interpolation} we are guaranteed to find a nonzero $Q \in \mathcal Q$ such that for all $b \ge 0$, the coefficient of $E^b$ in
\[
    Q\left(\alpha_i + T, y - \sum_{j=1}^d(-T)^j Y_j + TE, Y_1, \hdots, Y_d\right)
\]
is divisible by $T^{m-db}$. By substituting $E$ with $T^d E$, we have that the coefficient of $E^b$ in 
\[
    Q\left(\alpha_i + T, y - \sum_{j=1}^d(-T)^j Y_j + T^{d+1}E, Y_1, \hdots, Y_d\right)
\]
is divisible by $T^m$. Thus, (\ref{eq:Q-d+1}) is indeed divisible by $T^m$.

To apply the root-finding algorithm in \cref{thm:root-finding}, we can see the choice of the parameter $k$ is valid as $k > \lceil \eps^{-3/\theta}\rceil = d$. We need to verify that the field size $q$ is sufficiently large relative to $Q$. First, we bound $\deg_{Y_i} Q$ for all $i \in \{0, 1, \hdots, d\}$. Since for any $X^aY_0^{b_0}Y_1^{b_1}Y^c \in \mathcal Q$  we have that  $b_0 + b_1 + |c| \le B$, observe
\[
    \deg_{Y_i} Q \le B \le \left\lceil \frac{mA}{k-1}\right\rceil \le 2d^3 \eps \frac{n}{k} < 4\eps^{1-9/\theta} \frac{n}{k} \le q.
\]
where the last line follows from the fact that $q \ge 4\eps^{1-9/\theta} \frac{n}{k}$.

Furthermore, since for any $X^aY_0^{b_0}Y_1^{b_1}Y^c \in \mathcal Q$, we have that $a + (k-1)(b_0 + b_1 + |c|) < mA$, we have that
\[
    \deg_{1,k-1,k-2, \hdots, k-d-1} Q < mA < 2d^3 \eps n \le 4\eps^{1-9/\theta} \frac{n}{k} \cdot k \le q^2,
\]
where the last inequality follows from the fact that $q \ge n \ge k$ and $q \ge 4\eps^{1-9/\theta}\frac{n}{k}$.

To finish, we give a run-time analysis. The time needed by \cref{thm:root-finding} is $q^{O(d+1)}$ = $q^{O(\eps^{-3/\theta})}$, so it suffices to bound the time it takes to compute $Q$.

First, we upper-bound the number of monomials in $\mathcal Q$. For any $X^aY_0^{b_0}Y_1^{b_1}Y^c \in \mathcal Q$, we have that $\sum_{j=2}^d c_j \le W \le 2d^4 = O(\eps^{-12/\theta})$, so there are at most $d^{O(\eps^{-12/\theta})} \le q^{O(\eps^{-12/\theta})} $ choices for $c$. Since $b_0 + b_1 \le B$, there are at most $B^2 \le q^2$ choices for $b_0$ and $b_1$. Finally, there are at most $mA \le q^2$ choice for $a$. Thus, $\mathcal Q$ has at most $q^4 \cdot q^{O(\eps^{-12/\theta})}= q^{O(\eps^{-12/\theta})}$ monomials. To count the number of constraints imposed, observe when writing out the expansion (\ref{eq:Q-d+1}), each monomial of $Q$ can become up to $(mA+1)\binom{B+d+1}{d+1} \le q^2 \cdot q^{2(d+1)}$ monomials. Thus, the total number of equations imposed is at most $n \cdot q^{2d+4} \cdot q^{O(\eps^{-12/\theta})} = q^{O(\eps^{-12/\theta})}$. It is clear that each constraint can be computed in $q^{O(\eps^{-12/\theta})}$ time, so the total time it takes to set up the linear system defining $Q$ can be done in $q^{O(\eps^{-12/\theta})}$ time. Furthermore, it only takes $q^{O(\eps^{-12/\theta})}$ to find a nonzero $Q$ in the kernel of this linear system. Thus, altogether the runtime is $q^{O(\eps^{-12/\theta})}$ with at most $q^{O(\eps^{-3/\theta})}$ polynomials found.
\end{proof}

\section{Extending to All Rates}\label{sec:main-cor}
Here we state and prove the formal version of our all-rates main theorem, or the formal version of \cref{cor:all-rates}. 

\begin{corollary}[Formal Version of \Cref{cor:all-rates}]
\label{cor:all-rates-formal}
For every rate $R \in (0,1)$ and slack $\delta \in (0,1)$ there exists a constant $C >0$ such that the following
holds. Fix parameters $n\ge k\ge1$ satisfying
\[
\frac{k}{n}\le R,
\]
and let $q$ be any prime satisfying $q\ge Cn$. Then, for any distinct
evaluation points $\alpha_1,\ldots,\alpha_n\in\F_q$ and any received word
$(y_1,\ldots,y_n)\in\F_q^n$, there is a deterministic algorithm which
outputs, in $q^{O_{R,\delta}(1)}$ time, a list of all univariate
polynomials $P\in\F_q[X]$ of degree less than $k$ such that the agreement
between $(P(\alpha_1),\ldots,P(\alpha_n))$ and
$(y_1,\ldots,y_n)$
is at least $(R+\delta)n$. Furthermore, the length of this list is at most
$q^{O_{R,\delta}(1)}$.
\end{corollary}
\begin{proof}[Proof of \cref{cor:all-rates}]
Set
\[
\theta:=\frac{\delta}{2(R+\delta)},
\]
so that
\[
\frac{R}{1-\theta}<R+\delta.
\]
Fix a constant $\eta>0$ satisfying
\[
\eta<
\left(
\frac{\theta^3(1-\theta)}{768}
\right)^{\frac{5+\theta}{1-\theta}}
\qquad\text{and}\qquad
\eta<R+\delta.
\]
Let
\[
A:=\lceil(R+\delta)n\rceil,
\qquad
k'
:=
\max\left\{
k,\,
\left\lceil\eta^{-3/\theta}\right\rceil+1
\right\},
\]
and set
\[
N
:=
\max\left\{
n,\,
\left\lceil
\frac{k'}{(1-\theta)\eta}
\right\rceil
\right\}.
\]
By construction,
\[
k'>\left\lceil\eta^{-3/\theta}\right\rceil
\qquad\text{and}\qquad
\frac{k'}{N}\le(1-\theta)\eta.
\]

Moreover, since $k'\le k+O_{R,\delta}(1)\le Rn+O_{R,\delta}(1)$ and
$R/(1-\theta)<R+\delta$, for all sufficiently large $n$ we have $\eta N\le A$. Also $N=\Theta_{R,\delta}(n)$.

Choose $C=C(R,\delta)$ sufficiently large that whenever $q\ge Cn$,
\[
q\ge
\max\left\{
N,\,
4\eta^{1-9/\theta}\frac{N}{k'}
\right\}.
\]
Such a constant exists since $N=\Theta_{R,\delta}(n)$ and
$k'\ge1$.

Extend the original evaluation set to $N$ distinct points in $\F_q$, and
let $y'\in\F_q^N$ be the received word obtained by padding the original
word with $N-n$ zeros. We apply \cref{thm:main-formal} to $\mathsf{RS}_{N,k'}(\alpha_1,\ldots,\alpha_N)$
with parameters $(N,k',\eta,\theta,q)$ and received word $y'$.
The hypotheses of \cref{thm:main-formal} hold by the calculations above, so
we obtain a list $\mathcal L'$ containing every polynomial of degree less
than $k'$ having at least $\eta N$ agreements with $y'$.

Every polynomial of degree less than $k$ having at least $A$ agreements
with the original received word still has at least $A\ge\eta N$ agreements
with the padded word, and hence belongs to $\mathcal L'$. We therefore
prune $\mathcal L'$ by retaining only those polynomials $P$ satisfying
\[
\left|\{i\in[n]:P(\alpha_i)=y_i\}\right|\ge A.
\]
The resulting list is exactly the desired list decoding of the original word $y$.

Finally, since $\eta$ and $\theta$ depend only on $R,\delta$,
\cref{thm:main-formal} gives running time $q^{O_{R,\delta}(1)}$ and list
size $q^{O_{R,\delta}(1)}$, and the pruning step only adds polynomial time in
the size of the list.
\end{proof}
\printbibliography

\appendix

\section{A Proof of \cref{lem:simplexExpectation}}

We first recall the following lemma we seek to prove.

\begin{lemma}
Let $S=\operatorname{conv}\{v_0,\ldots,v_{d-1}\}$
be a $d-1$-dimensional simplex, and let $X$ be uniformly distributed on $S$. Then
\[
\mathbb{E}[X]=\frac{1}{d}\sum_{i=0}^{d-1} v_i.
\]
In particular, the expectation of $X$ is the centroid of $S$.
\end{lemma}

\begin{proof}

We start by considering the standard simplex
\[
\Delta_{d-1}
=
\left\{
(\lambda_0,\ldots,\lambda_{d-1})\in\mathbb{R}_{\geq 0}^{d}
:
\sum_{i=0}^{d-1}\lambda_i=1
\right\},
\]
and define the affine map
\[
F(\lambda_0,\ldots,\lambda_{d-1})
=
\sum_{i=0}^{d-1}\lambda_i v_i.
\]
Since the vertices $v_0,\ldots,v_{d-1}$ are affinely independent, $F$ is an affine bijection from $\Delta_{d-1}$ onto $S$.

The Jacobian of $F$, restricted to the affine hull of $\Delta_{d-1}$, is a fixed nonzero constant. Consequently, $F$ multiplies the volume of every measurable subset of $\Delta_{d-1}$ by the same factor. Thus, if $Z=(Z_0,\ldots,Z_{d-1})$ is uniformly distributed on $\Delta_{d-1}$, then $F(Z)$ is uniformly distributed on $S$. We may therefore write
\[
X\stackrel{\mathrm{d-1}}{=}\sum_{i=0}^{d-1}Z_i v_i.
\]

The uniform distribution on $\Delta_{d-1}$ is invariant under every
permutation of its coordinates. Hence
\[
\mathbb{E}[Z_0]
=
\mathbb{E}[Z_1]
=
\cdots
=
\mathbb{E}[Z_{d-1}].
\]
On the other hand, $\sum_{i=0}^{d-1}Z_i=1$ identically. Taking expectations gives
\[
\sum_{i=0}^{d-1}\mathbb{E}[Z_i]=1,
\]
and therefore
\[
\mathbb{E}[Z_i]=\frac{1}{d}
\qquad\text{for every }i.
\]
By linearity of expectation,
\[
\mathbb{E}[X]
=
\sum_{i=0}^{d-1}\mathbb{E}[Z_i]v_i
=
\frac{1}{d}\sum_{i=0}^{d-1} v_i,
\]
which is precisely the centroid of $S$.
\end{proof}
\end{document}